\documentclass[11pt,reqno]{amsart}
\usepackage[foot]{amsaddr}
\usepackage{graphicx}
\usepackage{amsthm,amssymb,amsmath,calc,eucal,ifthen}
\usepackage{amscd}
\usepackage[all,arc,curve,color,frame,matrix,arrow]{xy}
\usepackage{xcolor}
\usepackage{pgf,tikz,pgfplots}
\pgfplotsset{compat=1.15}
\usepackage{mathrsfs}
\usetikzlibrary{arrows}
\numberwithin{equation}{section}

\newtheorem{theorem}{Theorem}[section]
\newtheorem{lemma}[theorem]{Lemma}
\newtheorem{proposition}[theorem]{Proposition}

\theoremstyle{remark}
\newtheorem{remark}[theorem]{Remark}
\newtheorem{example}[theorem]{Example}

\newtheoremstyle{rmdefinition}{}{}{\upshape}{}{\bfseries}{.}{ }{}

\theoremstyle{rmdefinition}

\newcommand{\be}[1]{\begin{equation}\label{#1}}
\newcommand{\ee}{\end{equation}}
\newcommand{\beqa}{\begin{eqnarray}}
\newcommand{\eeqa}{\end{eqnarray}}

\newcounter{tmpc}
\newlength{\tmplenght}
\newlength{\tmplenghta}
\newlength{\tmplenghtb}
\newlength{\tmplenghtc}

\begin{document}

\title[Projective limits]{Projective limits in Euclidean quantum field theory, I: Free scalar fields}

\author{Svetoslav Zahariev}
\address{MEC Department, LaGuardia Community College of The City University of New York, 31-10 Thomson Ave., Long Island City, NY 11101, U.S.A.}
\email{szahariev@lagcc.cuny.edu}

\maketitle
\begin{abstract}
We present two constructions of projective systems of measures associated to discretizations of free scalar Euclidean quantum fields. The first one is obtained using only purely combinatorial data and applies to free massless scalar fields on polyhedral cell complexes. The second construction utilizes the suitably discretized covariance of the free massive scalar field on a compact Riemannian manifold.
\end{abstract}

\section{Introduction}
Recently, there has been considerable interest in obtaining inductive systems of operator algebras associated to lattice field theories in the context of algebraic quantum field theory (QFT), see e.g. \cite{LSW} and \cite{MMST}. In particular, \cite{MMST} establishes the existence of an inductive system of $C^*$-algebras associated to discretizations of the free massive scalar quantum field defined on a sequence of lattices whose spacing tends to zero. The limit of this inductive system is then identified with the corresponding continuum relativistic free quantum field. 

In this note, we tackle the analogous problem in the setting of Euclidean QFT which can be naturally formulated in terms of projective systems of probability measures, restricting our attention to free scalar fields, as in \cite{MMST}. One may hope that the usual discretization of the action corresponding to such a field via lattice Laplacians leads to a projective system of measures with respect to certain natural subdivision bonding maps. Indeed, our first main theorem asserts, for suitably chosen inner products on the spaces of 1-cochains, the existence of such a projective system in the case of a discretized free massless scalar field on a polyhedral cell complex of arbitrary dimension. This result may be regarded as an analogue of the known projective system of measures associated to a 2-dimensional lattice gauge theory, cf. \cite{KK} and \cite{VM}.

Unfortunately, identifying the limit of the latter projective system of measures with the continuum field appears rather difficult. Thus, we consider as an alternative a suitable discretization scheme for the covariance operator of the free Euclidean scalar field. Our second main result is based on this approach and may be stated as follows. Let $M$ be a closed Riemannian manifold and $\{K_{i}\}$ be a sequence of triangulations of $M$ (satisfying suitable technical conditions) whose mesh tends to 0. Then there exists a projective system of measures on the spaces of 0-cochains associated to $\{K_{i}\}$ whose limit may be naturally identified with the free massive scalar field on $M$. To construct this projective system, we utilize the framework of Whitney and de Rham maps developed in \cite{Do} which has already been recast in a Euclidean QFT setting in \cite{AZ}. We plan to apply similar methods in order to obtain projective systems of measures corresponding to higher dimensional lattice gauge theories in a subsequent work.

This article is organized as follows. In Section \ref{prelsec} we review several important facts pertaining to the theory of inductive and projective limits of Hilbert spaces, projective systems of measures and inductive systems of characteristic functionals. Section \ref{onedimsec} is dedicated to the case of a massless free scalar field on a polyhedron, while in Section \ref{rimmandsec} we establish our second main result as presented above, as well as an infinite volume limit of our construction.

\section{Preliminaries}\label{prelsec}
\subsection{Inductive and projective limits of Hilbert spaces}\label{indprojsubse}
In this section, we briefly discuss projective and inductive limits in $\mathbf{Hilb}_1$, the category whose objects are real Hilbert spaces and whose morphisms are linear contractions, i.e. maps of norm not exceeding 1 (see \cite[Section 2]{Gr} for details). We note that these limits should not be confused with the better known projective and inductive limits in the category of locally convex spaces.

Let $\{\mathcal{H}_i\}_{i=1}^{\infty}$ be a sequence of real finite dimensional inner product vector spaces and suppose we are given linear isometries $I_{ij}: \mathcal{H}_i \rightarrow \mathcal{H}_j$ for all $i\leq j$ satisfying $I_{ii}=\text{Id}$ and $I_{jk} I_{ij}=I_{ij}$ for all $i \leq j \leq k $.

Regarding $(\mathcal{H}_i,P_{ij})$, where $P_{ij}=I^*_{ij}$, as a projective system in $\mathbf{Hilb}_1$, its projective limit may be defined as follows. Denote by $\ell^{\infty}(\{H_{i}\})$ the Banach space 
of all sequences $\{h_{i}\}_{i \in \mathbb{N}}$, $h_i \in h_i$ such that
\begin{equation}\label{projnormsupde}
	\|\{h_{i}\}\|_{\infty}:=\sup_{i}\| h_i\|_{i}< \infty,
\end{equation}
where $\|\cdot \|_i$ stands for the norm induced by the inner product $\langle \cdot,\cdot \rangle_i$ on $H_{i}$. One sets
\begin{equation}\label{projlimband}
	\mathcal{H}_{\infty}:=
	\{\{h_i\} \in \ell^{\infty}(\{H_{i}\}): P_{ij}h_j=h_i, \forall i<j\}, \quad P_i(\{h_i\}):=h_i,
\end{equation}
and checks (cf. \cite[Theorem 2.3]{Gr}) that the pair $(\mathcal{H}_{\infty}, P_i)$,
where  $P_{i}:\mathcal{H}_{\infty} \rightarrow  H_{i}$,  is a projective limit of 
$(H_{i},P_{ij})$ in $\mathbf{Hilb}_1$. In particular, one has
$P_{ij}P_j=P_i$ for all $i<j$. Indeed, one sees  that the supremum in (\ref{projnormsupde}) becomes limit when applied to sequences in $\mathcal{H}_{\infty}$. It follows that the norm $\| \cdot \|_{\infty }$ on $\lim_{\leftarrow}\{H_i\}$ defined in  (\ref{projnormsupde}) satisfies the parallelogram identity since this is true for the norms $\|\cdot \|_i$, i.e. the Banach space $\mathcal{H}_{\infty}$ is in fact a Hilbert space.

It is also shown in \cite[Section 2]{Gr} that $\mathcal{H}_{\infty}$ may be identified with the inductive limit of $(\mathcal{H}_i,I_{ij})$ in $\mathbf{Hilb}_1$ as well, so that one obtains contractions $I_i:=P^{*}_{i}: \mathcal{H}_i \rightarrow \mathcal{H}_{\infty}$ satisfying $I_{ij}I_j=I_i$ for all $i<j$. Our assumption that $I_{ij}$ are isometries implies that so are the maps $I_i$. One has $P_i I_i=Id$ and $\cup_{i}\mathtt{Im}I_i$ is dense in $\mathcal{H}_{\infty}$ (see \cite[Section 4.2]{C}).

\begin{lemma}\label{leonequicof}
Let $f_i: \mathcal{H}_i \rightarrow \mathbb{C}$ be a sequence of continuous functions satisfying $f_i=f_jI_{ij}$ for all $j>i$. Assume  that the functions $f_iP_i$ defined on $\mathcal{H}_{\infty}$ are equicontinuous and pointwise uniformly bounded. Then there exists a unique continuous $f_{\infty}: \mathcal{H}_{\infty} \rightarrow \mathbb{C}$ satisfying
$f_i=f_{\infty}I_{i}$ for all $i$.
\end{lemma}
\begin{proof}
The consistency condition $f_i=f_jI_{ij}$ implies the existence of a unique function $f$ defined on $\cup_{i}\mathtt{Im}I_i$ whose restriction to $\mathtt{Im}I_i$ is $f_i$. Moreover, $f$ is continuous on $\cup_{i}\mathtt{Im}I_i$ equipped with the inductive limit topology (in the category of locally convex spaces) since its restrictions to $\mathtt{Im}I_i$ are continuous.
By our assumptions and the (generalized) Arzel\`a-Ascoli theorem the sequence $f_iP_i$ has a subnet converging to a continuous function $f_{\infty}$. Since $P_iI_i=Id$, $f_{\infty}$ coincides with $f$ on $\cup_{i}\mathtt{Im}I_i$.
\end{proof}

\subsection{Inductive limits of characteristic functionals}\label{Indlimfuncsec}
Consider a countable projective system $(X_i,P_{ij})$ in the category $\mathbf{Top}$ of topological spaces and continuous mappings, and a sequence of finite Radon measures $\mu_i$ on $X_i$. Recall that $(X_i,P_{ij},\mu_i)$ is called a {\em projective system of measures} if one has $P_{ij}(\mu_j)=\mu_i$ for all $j>i$, where $P_{ij}(\mu_j)$ stands for the image/pushforward of $\mu_j$ under $P_{ij}$. It is well-known (see e.g. \cite[Part I, Chapter I, \S 10]{Sz}) that in this situation there exists a unique finite Radon measure $\mu_{\infty}$ on $(X_{\infty},P_i)$, the projective limit of $(X_i,P_{ij})$ in $\mathbf{Top}$ such that $P_i(\mu_{\infty})=\mu_i$ for all $i$. The measure $\mu_{\infty}$ is called the projective limit of the sequence $\mu_i$.

In what follows we shall discuss, in the special case when $X_i$ are real finite dimensional inner product vector spaces, the dual notion of an inductive system of characteristic functionals. Let $(\mathcal{H}_i,I_{ij})$ be an inductive system in $\mathbf{Hilb}_1$ with $I_{ij}$ isometries, as in Section \ref{indprojsubse}. Further, let $S_i$ be a sequence of continuous positive definite functionals on $\mathcal{H}_i$ satisfying $S_jI_{ij}=S_i$ for all $i$. We shall call the triple $(\mathcal{H}_i,I_{ij},S_i)$ an {\em inductive system of characteristic functionals}.

To motivate the latter definition, we now briefly review the definition and basic properties of characteristic functionals of  finite measures. Let $\mathcal{H}$ be a real finite dimensional Hilbert space and let $\mu$ be a finite Borel measure on 
$\mathcal{H}$. The characteristic functional/Fourier transform of $\mu$ is given by 
$$S_{\mu}(h')=\int_{\mathcal{H}}e^{i\langle h, h' \rangle_{\mathcal{H}}} d\mu(h), \quad h' \in \mathcal{H}.$$
Recall that by Bochner's theorem ((see e.g. \cite[Theorem 7.13.1]{B})), every positive definite continuous functional on $\mathcal{H}$ is the characteristic functional of a finite Borel measure on  $\mathcal{H}$.

\begin{lemma}\label{funcpropcharfu}
Let $\mathcal{H}_1$ and $\mathcal{H}_2$ be two  real finite dimensional Hilbert spaces and $L: \mathcal{H}_1 \rightarrow \mathcal{H}_2$ be a linear map. Let $\mu_i$ be a finite Borel measure on $\mathcal{H}_i$, $i=1,2$. Then $L(\mu_1)=\mu_2$ if and only if $S_{\mu_1}L^*=S_{\mu_2}$.
\end{lemma}
\begin{proof}
We assume that $L(\mu_1)=\mu_2$ and applying the change of variables formula for pushforward measures, find
	\begin{multline*}
		S_{\mu_2}(h')=\int_{\mathcal{H}_2}e^{i\langle h', h\rangle_{\mathcal{H}_2}} dL(\mu_1)(h)=\int_{\mathcal{H}_1}e^{i\langle h', Lh\rangle_{\mathcal{H}_1}} d\mu_1(h)\\
		=\int_{\mathcal{H}_1}e^{i\langle L^*h', h\rangle_{\mathcal{H}_1}}d\mu_1(h)=
		S_{\mu_1}(L^*(h')),
	\end{multline*}
 for all  $h' \in \mathcal{H}_2$. Conversely, since a measure is determined by its characteristic functional (see e.g. \cite[Lemma 7.13.5]{B}), it follows that $S_{\mu_1}L^*=S_{\mu_2}$ implies $L(\mu_1)=\mu_2$.
\end{proof}

\begin{proposition}\label{prooncharfu} (a) Assume that the functionals $S_iP_i$ defined on $\mathcal{H}_{\infty}$ are pointwise uniformly bounded and equicontinous. Then there exists a unique continuous positive definite functional $S_{\infty}$ on $\mathcal{H}_{\infty}$ such that $S_{\infty}I_i=S_i$ for all $i$.
	
(b) The triple $(\mathcal{H}_i,I_{ij},S_{\mu_i})$ is an inductive system of characteristic functionals if and only if $(\mathcal{H}_i,I_{ij}^{*},\mu_i)$ is a  projective system of measures.
\end{proposition}
\begin{proof}
Part (a) follows immediately from Lemma \ref{leonequicof}, while part (b) follows from Lemma \ref{funcpropcharfu}.
\end{proof}

\begin{example}\label{gaumesexam}
Let $\mu_{A_i}$ be the centered Gaussian probability measure on $\mathcal{H}_i$ with covariance operator $A_i$ so that one has 
$$ S_{\mu_{A_i}}(h)=e^{-\frac{\langle A_i h, h\rangle_i}{2}}, 
	\quad h \in \mathcal{H}_i.$$
Then one easily sees that $(\mathcal{H}_i,I_{ij},S_{\mu_{A_i}})$ is an inductive system of characteristic functionals if and only if 
\begin{equation}\label{consgaucova}
I_{ij}^{*} A_j I_{ij}=A_i
\end{equation}
holds for all $j>i$.

Now assume in addition that the operators $A_i$ are uniformly bounded, i.e. there exists $K>0$ such that $\|A_i \|_i \leq K $ for all $i$. Then it is not hard to see that the functionals $S_{\mu_{A_i}}P_i$ are equicontinuous. Indeed, a short computation yields
$$ \int_{\mathcal{H}_i}|e^{i\langle h_1, h \rangle_{i}}
-e^{i\langle h_2, h \rangle_{i}} |^2 d\mu_{A_i}(h)=
2(1-e^{-\|A^{1/2}_{i}(h_1-h_2)\|_{i}^{2}/2})$$
for all $h_1,h_2 \in \mathcal{H}_i$. It follows that 
$$| S_{\mu_{A_i}}(P_i(h_1))-S_{\mu_{A_i}}(P_i(h_2))| \leq 
\sqrt{2}(1-e^{-\|A^{1/2}\|_{i}\|h_1-h_2\|^{2}/2})^{1/2}$$
for all $h_1,h_2 \in \mathcal{H}_{\infty}$, which implies the desired equicontinuity. Thus by Proposition \ref{prooncharfu}(a) there exists a unique continuous positive definite functional $S_{\infty}$ on $\mathcal{H}_{\infty}$ such that $S_{\infty}I_i=S_{\mu_{A_i}}$ for all $i$.
\end{example}

\section{The massless free scalar field and discrete Laplacians}\label{onedimsec}

\subsection{Basic notions} 

Let $\{K_i\}_{i=0}^{\infty}$ be a sequence of $d$-dimensional finite polyhedral cell complexes embedded in $\mathbb{R}^d$ such that $K_i$ is a subdivision of $K_{i-1}$ for all $i$. 

We write $C^{k}(K_i)$ for the space of (oriented) real-valued $k$-cochains on $K_i$ and denote by $d_i:C^{0}(K_i) \rightarrow :C^{1}(K_i)$ the usual coboundary operator (see e.g. \cite[\S 42]{Mu}). Further, we write  $C^{0,0}(K_i)$ for the subspace of $C^{0}(K_i)$ consisting of all cochains vanishing on the vertices belonging to the (topological) boundary of $K_i$ regarded as a subspace of $\mathbb{R}^d$.

Given inner products $\langle \cdot, \cdot \rangle_{k,i}$ on $C^{k}(K_i)$ for $k=0,1$, we define positive discrete Laplacians on $C^{0}(K_i)$ and $C^{0,0}(K_i)$ via setting
$$ \Delta_i=d_id^*_i  ,\quad \Delta_{i,0}=d_{i,0}d^*_{i,0}.$$

From now on we assume that the cell complexes $K_i$ are connected manifolds with non-empty boundary. It follows that the restricted coboundary operators $d_{i,0}$ are injective and the Dirichlet Laplacians $\Delta_{i,0}$ are invertible, which in turn allows us to define centered Gaussian probability measures
\begin{equation}\label{frfmedes}
d\mu_{i,0}(c)=z_i e^{-\frac{1}{2}\langle c, \Delta_{i,0} c\rangle _{0,i}} d c
\end{equation}
on $C^{0,0}(K_i)$, where $z_i$ is a normalization constant and $dc$ stands for the Lebesgue measure on $C^{0,0}(K_i)$. 
We note that $\mu_{i,0}$ has covariance  $\Delta^{-1}_{i,0}$.

\begin{example}\label{cubexam}
We take $K_0$ to be the $d$-dimensional polyhedral cell complex embedded in $\mathbb{R}^d$ whose single $d$-cell is the cube $[-1,1]^d$. For each positive integer $i$ we define inductively $K_i$ to be the polyhedral cell complex obtained from $K_{i-1}$ by subdividing each $d$-cell in $K_{i-1}$ into $2^d$ congruent cubes having side equal to $1/2^{i-1}$ Thus $K_i$ is a $d$-dimensional cubical lattice with mesh $1/2^{i-1}$.

We introduce the following inner products on $C^{k}(K_i)$. First, we write $\langle \cdot, \cdot \rangle_{k,i,0}$ for the canonical inner product on $C^{k}(K_i)$ with respect to which the $k$-cells form orthonormal basis. Second, we set
\begin{equation}\label{defscalinnp}
	\langle c_1, c_2 \rangle_{k,i}= 2^{(i-1)(2k-d)} \langle c_1, c_2 \rangle_{k,i,0}, \quad c_1,c_2 \in 
	C^{k}(K_i).
\end{equation}
For the motivation behind the scaling factor $2^{(i-1)(2k-d)}$ appearing above, the reader is referred to \cite[Lemma 7.22]{DP}. It is easy to see that in this case the discrete Laplacians introduced above coincide with the standard finite difference Laplacians on cubical lattices as defined e.g. in \cite[Section 9.5]{GJ}, hence the measures (\ref{frfmedes}) represent the free massless Euclidean scalar field on the lattices $K_i$.
\end{example}

\subsection{The projective system of measures and the continuum limit} 

We now assume that we are given linear surjections $$P^k_{i}: C^{k}(K_i) \rightarrow C^{k}(K_{i-1})$$ for all $i>0$ and $k=0,1$, such that 
$P^0_{i}$ maps $C^{0,0}(K_i)$ into $C^{0,0}(K_{i-1})$ and one has 
\begin{equation}\label{cochamapco}
	P^1_{i}d_{i,0} =d_{i-1,0} P^0_{i}, \quad i>0.
\end{equation}
In the special case of Example \ref{cubexam}, one can take $P^0_{i}$ to be  the map given by restriction of 0-chains and $P^1_{i}$ to be the map that sends a $1$-cochain $c \in C^{1}(K_i)$ to the cochain whose value at a 1-cell (edge) $e$ in $K_{i-1}$ is the sum of the values of $c$ at the two sub-cells of $e$ belonging to $K_{i}$. It is well-known (and easy to check) that $P^0_{i}$ and $P^1_{i}$ form a cochain map, i.e. (\ref{cochamapco}) is satisfied. Natural subdivision maps obeying (\ref{cochamapco}) exist also in the case when all $K_i$ are simplicial complexes (see \cite[Chapter 2, \S 17]{Mu}).

We observe that by (\ref{cochamapco}),  $P^1_{i}$ maps $\mathtt{Im}\hspace{1pt}d_{i,0}$ to  $\mathtt{Im}\hspace{1pt}d_{i-1,0}$.
We shall now construct new, ``renormalized'' inner products $\langle \cdot, \cdot \rangle_{1,i}^r$ on $\mathtt{Im}\hspace{1pt}d_{i,0}$ out of the inner products $\langle \cdot, \cdot \rangle_{1,i}$  turning the maps $P^1_{i}$ into co-isometries. We set 
$$ \langle \cdot, \cdot \rangle_{1,0}^r=\langle \cdot, \cdot \rangle_{1,0}$$
and, using the decomposition
\begin{equation}\label{basiorthde}
\mathtt{Im}\hspace{1pt}d_{i,0}=\mathtt{Im}\hspace{1pt}(P^1_{i})^* \oplus (\mathtt{Im}\hspace{1pt}(P^1_{i})^*)^{\perp},
\end{equation} 
for every $i>0$ inductively define $\langle \cdot, \cdot \rangle_{1,i}^r$   to be the pullback of $\langle \cdot, \cdot \rangle_{1,i-1}^r$ via $(P^1_{i})$ on $\mathtt{Im}\hspace{1pt}(P^1_{i})^*$ and to coincide with $\langle \cdot, \cdot \rangle_{1,i}$ on $(\mathtt{Im}\hspace{1pt}(P^1_{i})^*)^{\perp}$.
Finally, we declare the two summands in (\ref{basiorthde}) to be orthogonal with respect to  $\langle \cdot, \cdot \rangle_{1,i}^r$. (Above, the adjoints and orthogonal complements are taken with respect to the inner products $\langle \cdot, \cdot \rangle_{1,i}$ restricted to $\mathtt{Im}\hspace{1pt}d_{i,0}$.)

Next we define projective systems of $0$ and $1$-cochains via setting
$$P^k_{ij}:=P^k_{i} P^k_{i-1} \cdots  P^k_{j+1}, \quad i>j, \\ k=0,1,$$
and note that 
\begin{equation}\label{cochamapcoij}
P^1_{ij}d_{i,0} =d_{j,0} P^0_{ij}, \quad i>j.
\end{equation}
Finally, we introduce the ``renormalized'' version of the free massless  scalar field measures (\ref{frfmedes}):
\begin{equation}\label{frfmedestwo}
	d\mu_{i,r}(c)=z_{i,r} e^{-\frac{1}{2}\langle c, \Delta_{i,r} c\rangle _{0,i}} d c,
\end{equation}
 where $z_{i,r}$ is a normalization constant, $dc$ stands for the Lebesgue measure on $C^{0,0}(K_i)$, and  $\Delta_{i,r}=d_{i,0}d^{*,r}_{i,0}$, where $*,r$ stands for the adjoint taken with respect to the inner product $\langle \cdot, \cdot \rangle_{1,i}^r$.

We shall make use of the following simple fact.
\begin{lemma}\label{chmeasulem}
Let $\mathcal{V}_1$ and $\mathcal{V}_2$ be two finite dimensional real vector spaces and let $L$ be a linear isomorphism. Let $f: \mathcal{V}_2 \to \mathbb{R}$ be a positive continuous integrable function.

Define Borel probability measures on $\mathcal{V}_i$
via 
\[ d\mu_1(v_1) = N_1 f(L(v_1)) \, dv_1, \]
\[ d\mu_2(v_2) = N_2 f(v_2) \, dv_2, \]
where $dv_i$ stands for the Lebesgue measure on $\mathcal{V}_i$ and $N_i$ is a normalization constant. Then $\mu_1 = L^{-1}(\mu_2)$.
\end{lemma}
\begin{proof}
This follows easily from the change of variables formula for pushforward measures.
\end{proof}
\begin{theorem}\label{mainprojsythe}
The triple $(C^{0,0}(K_i),P^0_{ij},\mu_{i,r})$ is a projective system of measures.
\end{theorem}
\begin{proof}
We write $d_{i,0}^{-1}$ for the inverse of the linear isomorphism from $C^{0,0}(K_i)$ to $\mathtt{Im}\hspace{1pt}d_{i,0}$ given by $d_{i,0}$ and $\nu_i$ for the centered Gaussian probability measure on $\mathtt{Im}\hspace{1pt}d_{i,0}$ with characteristic functional
$$S_{\nu_i} (c)=e^{-\frac{1}{2}\langle c, c \rangle_{1,i}^r},\quad c \in \mathtt{Im}\hspace{1pt}d_{i,0}.$$
Using Lemma \ref{chmeasulem}, we conclude that $\mu_{i,r}=d_{i,0}^{-1}(\nu_i)$, therefore one has
\begin{equation}\label{charconide}
S_{\mu_{i,r}} =S_{\nu_i}(d_{i,0}^{*,r})^{-1}
\end{equation}
by Lemma \ref{funcpropcharfu}. As observed above, the maps $(P^1_{ij})^{*,r}$ are isometries, hence (cf. Example \ref{gaumesexam}) $(\mathtt{Im}\hspace{1pt}d_{i,0},(P^1_{ij})^{*,r}, S_{\nu_i})$ is an inductive system of characteristic functionals, i.e. one has
\begin{equation}\label{twocompcharco}
	S_{\nu_{i}}(P^1_{ij})^{*,r} = S_{\nu_{j}}.
\end{equation}
Further, we infer from (\ref{cochamapcoij}) that
\begin{equation}\label{conjchainma}
(d_{i,0}^{*})^{-1}(P^0_{ij})^{*,r} 
= (P^1_{ij})^{*,r}(d_{j,0}^{*})^{-1}, \quad i>j.
\end{equation}
It follows immediately from (\ref{charconide}), (\ref{twocompcharco}) and (\ref{conjchainma}) that
$ S_{\mu_{i,r}}(P^0_{ij})^{*,r} = S_{\mu_{j,r}}$, which by Lemma \ref{funcpropcharfu} implies that 
$P^0_{ij}(\mu_{i,r})=\mu_{j,r}$.
\end{proof}
\begin{remark}
(1) By the general existence theorem mentioned in the beginning of Section \ref{Indlimfuncsec}, the projective system of measures $(C^{0,0}(K_i),P^0_{ij},\mu_{i,r})$ has a limit measure $\mu_{\infty,r}$ defined on $C^{0,0}_{\infty}$, the limit of the projective system $(C^{0,0}(K_i),P^0_{ij})$.  It remains unclear whether and in what sense $\mu_{\infty,r}$, in the special case described in Example \ref{cubexam}, may be identified with the continuum massless free field which exists as a Gaussian Borel probability measure on a suitable space of distributions on $[-1,1]^d$.

(2) Consider the special case of Example \ref{cubexam} and suppose that  $d=1$,  so that the coboundary operators $d_{i,0}$ are surjective. In this situation 
no renormalized inner products are needed, as the maps $P^1_{ij}$ are already isometries (up to normalization factors) with respect to the inner products $\langle \cdot, \cdot \rangle_{1,i}$. Moreover, the fact that $(C^{0,0}(K_i),P^0_{ij},\mu_{i,r})$ form a projective system of measures may be established by direct computation, either by utilizing the convolution semigroup property of the heat kernel on $\mathbb{R}$, or using the explicit formula for the Green operator $\Delta^{-1}_{i,0}$ obtained in \cite[Theorem 3]{CY}. The heat kernel convolution semigroup method has already been employed in \cite{VM}  to show the existence of a projective system of measures associated to 2-dimensional lattice gauge theory.

(3) Let us write $\mathtt{Im}\hspace{1pt}d_{\infty,0}$  for the projective limit of $(\mathtt{Im}\hspace{1pt}d_{i,0},P^1_{ij})$ in the category of locally convex spaces and $\nu_{\infty}$ for the projective limit of the sequence of measures $\nu_i$ considered in the proof of Theorem \ref{mainprojsythe}. By properties of projective limits the maps $d_{i,0}^{-1}$  induce a continuous map $d_{\infty,0}^{-1}: \mathtt{Im}\hspace{1pt}d_{\infty,0} \to C^{0,0}_{\infty}$ and one has $\mu_{\infty,r}=d_{\infty,0}^{-1}(\nu_{\infty})$. This alternative description of $\mu_{\infty,r}$ suggests a different construction of a measure $\tilde{\mu}_{\infty,r}$ on $C^{0,0}_{\infty}$ representing the massless free field on $[-1,1]^d$ which uses only the initial inner products $\langle \cdot, \cdot \rangle_{1,i}$. Consider the inductive system $((\mathtt{Im}\hspace{1pt}d_{i,0},\langle \cdot, \cdot \rangle_{1,i}),(P^1_{ij})^{*})$ and denote by $\mathtt{Im}\hspace{1pt}d^{\hspace{1pt}\infty}_{0}$ and $\mathtt{Im}\hspace{1pt}d^{H}_{0}$ its inductive limits in the categories of locally convex spaces and Hilbert spaces and linear contractions, respectively. Identifying the dual of $\mathtt{Im}\hspace{1pt}d^{\hspace{1pt}\infty}_{0}$ with $\mathtt{Im}\hspace{1pt}d_{\infty,0}$, we obtain via the Minlos theorem a centered Gaussian probability measure $\tilde{\nu}_{\infty}$ on $\mathtt{Im}\hspace{1pt}d_{\infty,0}$ with characteristic functional $e^{-\frac{1}{2}\langle \cdot , \cdot \rangle_{1,\infty}}$, where $\langle \cdot , \cdot \rangle_{1,\infty}$ is the inner product on $\mathtt{Im}\hspace{1pt}d^{H}_{0}$, and set $\tilde{\mu}_{\infty,r}=d_{\infty,0}^{-1}(\tilde{\nu}_{\infty})$. Clearly $\tilde{\mu}_{\infty,r}=\mu_{\infty,r}$ when $d=1$ but in higher dimensions these two measures do not necessarily coincide.
\end{remark}

\section{The massive free scalar field and discretized covariances}\label{rimmandsec}
\subsection{The de Rham and Whitney maps}\label{derhamwhisec}
Let $M$ be a closed oriented Riemannian $d$-dimensional manifold and consider a sequence $\{K_{i}\}_{i=1}^{\infty}$ of smooth triangulations of $M$ whose fullness is bounded away from 0. We denote the mesh (diameter) of $K_{i}$ by $h_{i}$ and assume that $h_{i}\rightarrow 0$  as $i \rightarrow \infty$. (For the definitions of fullness and mesh, of which we shall not make use explicitly, we refer the reader to \cite[Section 2]{DP}.) Assume further that $K_{i+1}$ is a subdivision of $K_{i}$ for each $i$. 

In what follows, we write $C^n(K_{i})$ for the space of the real-valued $n$-cochains on the simplicial complex $K_{i}$ and $C(M)$ for the continuous real-valued functions on $M$. For every $i$ we define a linear map $R_i: C(M) \rightarrow C^0(K_{i})$ via
$$ R_i(f)(v)=f(v), \quad f \in C(M),$$
where $v$ is any vertex of $K_{i}$. Further, we define a mapping $W_i: C^0(K_{i}) \rightarrow C(M)$ by setting
\begin{equation}\label{whimapdefi}
 W_i(v)=\mu_v,
\end{equation} 
where $\mu_v$ stands for the barycentric coordinate function of the vertex $v$ in $K_{i}$, and extending by linearity. The maps $R_i$ and $W_i$ are special case of the so called {\em de Rham} and {\em Whitney maps} acting between cochains and piecewise smooth differential forms which were introduced in \cite{Do}.

We note that $\mathtt{Im}W_i$ consists of piecewise affine functions and summarize some of the basic properties of $R_i$ and $W_i$ in the following proposition.

\begin{proposition}\label{whitmapprope} (a) $R_{i}W_{i}=\text{Id}$.
	
	(b) $\lim_{i \rightarrow \infty} \|W_{i}R_{i}f-f\|_{L^2(M)}=0$ for every $f \in C^{\infty}(M)$.
	
	(c) $\cup_{i}\mathtt{Im}W_i$ is dense in $L^2(M)$.
	\end{proposition}
	\begin{proof}
Part (a) is clear from the definitions, while part (b) is a special case of the approximation results in \cite[Section 2]{DP}. Part (c) follows from part (b) since smooth functions are dense in $L^2(M)$.
	\end{proof}
Given $a,b \in C^0(K_{i})$ we call
\begin{equation}\label{whitinnerp}
\langle c_1,c_2\rangle_{W,i} = \langle W_{i}c_1, W_{i}c_2\rangle_{L^2(M)}
\end{equation}
the {\em Whitney inner product} of the 0-cochains $c_1$ and $c_2$. For all $j\geq i$ we define linear maps
$I_{ij}^{W}:  C^0(K_{i}) \rightarrow C^0(K_{j})$ by setting
\begin{equation}\label{imapdefi}
I_{ij}^{W}= R_j W_i.
\end{equation} 

\begin{proposition}\label{indusyswhi}
(a) The pair $( C^0(K_{i}), I_{ij}^{W})$, where $C^0(K_{i})$ is equipped with the Whitney inner product, is an inductive system in $\mathbf{Hilb}_1$  and the maps $I_{ij}^{W}$ are isometries. 

(b) The pair $( L^2(M), W_i)$ is a coherent system for $( C^0(K_{i}), I_{ij}^{W})$, i.e. 
the maps $W_i:  C^0(K_{i}) \rightarrow L^2(M)$ are isometries and one has 
$W_jI_{ij}^{W}=W_i$ for every $i$.

\end{proposition}	
\begin{proof}
We first observe that one has
\begin{equation}\label{mainwrw}
W_jR_jW_i=W_i
\end{equation}
whenever $j>i$. Indeed, for any $c \in  C^0(K_{i})$ both $W_ic$ and $W_jR_jW_ic$ are continuous functions on $M$ whose restrictions to each 
 $d$-simplex in $K_{j}$ are affine and take the same values on the vertices of the simplex. Since an affine function on a simplex is uniquely determined by its values on the vertices of the simplex (see \cite{Do2} for a generalization of this fact), we conclude that $W_ic$ and $W_jR_jW_ic$ coincide, thereby proving part (b).
 
Further, using (\ref{mainwrw}), we find that
 $$I_{kj}^{W}I_{ij}^{W}= R_k W_jR_j W_i =R_k  W_i =I_{ik}^{W},$$
 whenever $k>j>i$. Similarly, again by (\ref{mainwrw}), one has
 $$ \langle I_{ij}^{W}c_1,I_{ij}^{W}c_2\rangle_{W,j} = \langle W_{j}R_j W_ic_1,  W_{j}R_j W_ic_1c_2\rangle_{L^2(M)}= \langle c_1,c_2\rangle_{W,i} $$
for all $c_1,c_2 \in C^0(K_{i})$, hence the maps $I_{ij}^{W}$ are isometric and part (a) is established. 
\end{proof}

We note that by Proposition \ref{whitmapprope}(c) and the construction of inductive limits in $\mathbf{Hilb}_1$ there is a natural isometric isomorphism between the inductive limit Hilbert space of the system $( C^0(K_{i}), I_{ij}^{W})$, which we shall denote by $C^0_{\infty}(M)$, and $L^2(M)$.
\begin{remark}
Utilizing the de Rham and Whitney maps for higher degree cochains, one can define Whitney inner product on $C^n(K_{i})$ as in (\ref{whitinnerp}) and maps $I_{ij}^{W,n}:  C^n(K_{i}) \rightarrow C^n(K_{j})$ as in (\ref{imapdefi}) for $n>0$. However, $I_{ij}^{W,n}$ do not form an inductive system anymore since (\ref{mainwrw}) does not hold on cochains of positive degree. In view of the approximation property from Proposition \ref{whitmapprope}(b), it is natural to conjecture that $I_{ij}^{W,n}$ define a {\em soft inductive system} of Banach spaces, a weaker notion recently introduced in \cite[Section 3]{LSW}.
\end{remark}
\subsection{The inductive system of characteristic functionals}
We denote by $\Delta$ the positive Laplacian on $M$ acting on real functions and recall that $A^{\Delta}:=(1+\Delta)^{-1}$ is a bounded positive operator on
$L^2(M)$. By Minlos' theorem, one obtains a centered Gaussian Borel probability measure on $\mathcal{D}'(M)$, the distributions on $M$, with covariance $A^{\Delta}$ which represents the free scalar Euclidean quantum field of mass 1 on $M$ (cf. \cite[Chapter 6]{GJ}.) 

We define operators $A_i^{\Delta}$ on $C^0(K_{i})$ via

\begin{equation}\label{aioperdefj}
A_i^{\Delta}=W_i^{*}A^{\Delta}W_i
\end{equation}
where the adjoint is taken with respect to the Whitney inner product, and characteristic functionals
\begin{equation}\label{charfuncont}
S_{A^{\Delta}}(f)=e^{-(1/2)\langle A^{\Delta} f, f\rangle_{L^2(M)}}   ,\quad  f \in L^2(M),
\end{equation}
\begin{equation}\label{charfundiscr}
 S_{A^{\Delta}_{i}}(c)=e^{-(1/2)\langle A^{\Delta}_{i}c, c\rangle_{W,i}}  ,\quad  c\in C^0(K_{i}).
\end{equation}
\begin{theorem}\label{theoremonman}
The triple $( C^0(K_{i}), I_{ij}^{W}, S_{A^{\Delta}_{i}})$ is an inductive system of characteristic functionals, the inductive limit functional exists and may be identified with $S_{A^{\Delta}}$ via the isomorphism between $C^0_{\infty}(M)$ and $L^2(M)$.
\end{theorem}
\begin{proof}
Using Proposition \ref{indusyswhi}(b) we find
\begin{multline*}
(I_{ij}^{W})^{*}A_j^{\Delta} I_{ij}^{W} = (I_{ij}^{W})^{*}W_j^{*}A^{\Delta}W_j I_{ij}^{W}=
(W_jI_{ij}^{W})^{*}A^{\Delta}W_j I_{ij}^{W}  = W_iA^{\Delta}W_i A_i^{\Delta},
\end{multline*}
hence (\ref{consgaucova}) holds and we conclude that $( C^0(K_{i}), I_{ij}^{W}, S_{A^{\Delta}_{i}})$ is an  inductive system of characteristic functionals. Further, since the operators $A_i^{\Delta}$ are uniformly bounded, we see, using  Example (\ref{gaumesexam}) and Proposition \ref{prooncharfu}(a)  that a limit functional exists. The uniqueness of the latter limit allows us to identify it with $S_{A^{\Delta}}$. 
\end{proof}

\subsection{The infinite volume limit}

We now consider a sequence of triangulations ${K_i}$ of $\mathbb{R}^d$ that satisfy the assumptions stated in the beginning of Section \ref{derhamwhisec}. For each $i$ we choose a finite $d$-dimensional subcomplex $K_i^0$ of $K_i$ such that $K_i^0\subset K_{i+1}^0$ and 
$\cup_{i=1}^{\infty}K_i^0=\mathbb{R}^d$. For example, we can take $K_i^0$ to form a nested sequence of cubes.

We write $C^0(K_i^0)$ for the real valued $0$-cochains on $K_i^0$ that vanish on the boundary $\partial K_i^0$ of $K_i^0$ and $C_c(\mathbb{R}^d)$ for the real compactly supported continuous functions on $\mathbb{R}^d$. For every $i$ we define the de Rham map $R_i: C_c(\mathbb{R}^d) \rightarrow C^0(K_{i})$ to be the evaluation map at all points in $K_i^0 \setminus \partial K_i^0$ and 0 on $\partial K_i^0$. We note that the Whitney map $W_i$ given by (\ref{whimapdefi}) is well-defined as a map from $C^0(K_{i})$ to $C_c(\mathbb{R}^d)$. 

Thus we can define Whitney inner product on $C^0(K_i^0)$ as in (\ref{whitinnerp}) and bonding maps $I_{ij}^{W,0}$ given by (\ref{imapdefi}).
One checks exactly as in proof of Proposition \ref{indusyswhi} that the pair $( C^0(K_{i}^0), I_{ij}^{W,0})$, where $C^0(K_{i}^0)$ is equipped with the Whitney inner product, is an inductive system in $\mathbf{Hilb}_1$  and the maps $I_{ij}^{W,0}$ are isometries. 

We define a bounded positive operator on $L^2(\mathbb{R}^d)$ via
$$ A^d=(1+\Delta_d)^{-1},$$
where $\Delta_d$ is the positive Laplacian on $\mathbb{R}^d$. Further, we  consider operators $A^d_i$ on $C^0(K_{i}^0)$ given by (\ref{aioperdefj}) and characteristic functionals $S_{ A^d}$ and $S_{A^d_i}$ given by (\ref{charfuncont}) and (\ref{charfundiscr}), respectively.

We note that the approximation property in Proposition \ref{whitmapprope}(b) continues to hold for compactly supported smooth functions on $\mathbb{R}^d$, hence $\cup_{i}\mathtt{Im}W_i$ is dense in $L^2(\mathbb{R}^d)$ and we conclude that the inductive limit space of the inductive system $( C^0(K_{i}^0), I_{ij}^{W,0})$ in the category of Hilbert spaces and contractions may be identified with $L^2(\mathbb{R}^d)$. Thus Theorem \ref{theoremonman} may be reformulated as follows.

\begin{theorem}
	The triple $( C^0(K_{i}), I_{ij}^{W,0}, S_{A^d_i})$ is an inductive system of characteristic functionals, the inductive limit functional exists and may be identified with $S_{A^{d}}$.
\end{theorem}

\end{document}